\documentclass[a4,12pt]{article}%
\usepackage[cmex10]{amsmath}
\usepackage{amsmath,amsthm,amssymb}
\usepackage{epic, eepic}
\usepackage{graphicx}
\usepackage{amsfonts}
\usepackage{amssymb}%
\providecommand{\U}[1]{\protect\rule{.1in}{.1in}}
\newtheorem{theorem}{Theorem}[section]
\newtheorem{theorem2}{Theorem}[section]

\newtheorem{corollary}{Corollary}[section]

\newtheorem{definition}{Definition}[section]

\newtheorem{idea memo}{Idea Memo}[section]
\newtheorem{lemma}{Lemma}[section]

\newtheorem{remark}{Remark}[section]

\begin{document}

\title{Local state and sector theory in local quantum physics}
\author{Izumi Ojima\thanks{ojima@kurims.kyoto-u.ac.jp}\\
Shimosakamoto, Otsu, Shiga 520-0105, Japan \vspace{2mm}\\
Kazuya Okamura\thanks{okamura@math.cm.is.nagoya-u.ac.jp}\\
Graduate School of Information Science, Nagoya University\\
Chikusa-ku, Nagoya 464-8601, Japan\vspace{2mm}\\
Hayato Saigo\thanks{h\_saigoh@nagahama-i-bio.ac.jp}\\
Nagahama Institute of Bio-Science and Technology\\
Nagahama 526-0829, Japan}

\maketitle

\begin{abstract}
We define a new concept of local states in the framework of algebraic quantum field theory (AQFT).
Local states are a natural generalization of states and give
a clear vision of localization in the context of QFT.
In terms of them, we can find a condition from which follows automatically the famous
DHR selection criterion in DHR-DR theory.
As a result, we can understand the condition as consequences of
physically natural state preparations in vacuum backgrounds.
Furthermore, a theory of orthogonal decomposition of completely positive (CP) maps is developed.
It unifies a theory of orthogonal decomposition of states and order structure theory of CP maps.
By using it, localized version of sectors is formulated,
which gives sector theory for local states with respect to general reference representations.
\end{abstract}

\begin{flushleft}
\textbf{Keywords.}$\quad$ local state, sector theory, CP-measure space
\end{flushleft}
\textbf{Mathematics Subject Classification (2010)} 81T05, 46L53, 81P16
\begin{flushleft}

\end{flushleft}

%<-this % stops a space

\section{Introduction}\label{se:1}

As a mathematical foundation of quantum field theory (QFT), algebraic quantum field theory (AQFT) proposed by
Haag and Kastler \cite{HaagKastler} and by Araki has
made great contributions to our understanding of QFT. In particular,
sector theory developed by Doplicher, Haag and Roberts (DHR, for short)
\cite{DHR1,DHR2}
and by Doplicher and Roberts (DR, for short) \cite{DR89no1,DR89no2,DR90}
succeeded in treating local excitations as deviations from a vacuum and gave a contribution to clarifying
the observational origin of statistics of quantum fields and that of internal symmetries.

In AQFT, we begin with a local net $\{\mathcal{A}(\mathcal{O})\}_{\mathcal{O}\in\mathcal{K}}$
consisting of $W^\ast$-algebras
$\mathcal{A}(\mathcal{O})$ of observables for each bounded region $\mathcal{O}\in
\mathcal{K}$ of four dimensional Minkowski space $M_4=\mathbb{R}^4$.
In this mathematical formulation, we
analyze quantum fields by means of families of observables which are measurable in bounded
space-time regions. On the other hand, for the purpose of the description of physical situations and
of experimental settings, we use states defined as normalized positive linear functionals on the global
algebra
\begin{equation}
\mathcal{A}=\overline{\bigcup_{\mathcal{O}\in\mathcal{K}} \mathcal{A}(\mathcal{O})}^{\Vert\cdot\Vert}
\end{equation}
of a given local net $\{\mathcal{A}(\mathcal{O})\}_{\mathcal{O}\in\mathcal{K}}$. Typical examples are
vacuum states $\omega_0$ and KMS states $\omega_\beta$ where $\beta>0$.
However, several investigations (see \cite{Kitajima09} and references therein)
revealed a difficulty establishing a criterion of localizations
in bounded space-time regions in terms of states on the global algebra.
For instance, it is shown in \cite{Kitajima09}
that there is no ``detection operator" in any ``localized states" defined on the global algebra.
Here, a detection operator is defined by a projection of a local algebra, and
a state on the global algebra is said to be localized in a
bounded region $\mathcal{O}$ if it is identical with a vacuum state
on the algebra of observables on the causal complement $\mathcal{O}^\prime$ of $\mathcal{O}$.
Thus we cannot expect that localizations in bounded space-time regions are described (or characterized)
by methods similar to the quantum mechanical method.
In contrast to the standard formulation explained above, we can specify in the real world
physical situations and experimental settings in bounded regions. Thus it is natural for us 
to consider that,
for each bounded region $\mathcal{O}$ where measurements and physical operations are performed,
(normal) states on $\mathcal{A}(\mathcal{O})$ are specified, and that states on $\mathcal{A}$
are nothing but an ideal concept realized in the limit of $\mathcal{O}$ tending to $\mathbb{R}^4$.
Along the line of this observation, we can realize that the use of the presheaf structure of
$\{\mathcal{A}(\mathcal{O})_\ast\}_{\mathcal{O}\in\mathcal{K}}$ enables us to specify localized states
and to define states at a single space-time point $x\in\mathbb{R}^4$ as germs \cite{FredeHaag87,HaagOjima}.
As is given in \cite{HaagOjima}, this viewpoint is strongly connected with operator product expansion (OPE).
Bostelmann \cite{Bostelmann,Bostelmann05no1,Bostelmann05no2}
then justified both quantum fields defined at a single space-time point
and OPE under a condition similar to the nuclearity condition
in the context of AQFT, and Buchholz, Ojima and Roos \cite{BOR01}
discussed observables at a single space-time point in order to define and characterize
thermal nonequilibrium local states.

The concept of ``local states" defined in section \ref{se:3}
is a new direction for developing the idea of \cite{FredeHaag87,HaagOjima}.
Local states first introduced by Werner \cite{Werner} are a kind of quantum operations which were
proposed by Haag and Kastler \cite{HaagKastler} and described by completely positive (CP) maps on $\mathcal{A}$.
They play the role of states on $\mathcal{A}(\mathcal{O})$ for some bounded region $\mathcal{O}$
and of identity maps on $\mathcal{A}((\widetilde{\mathcal{O}})^\prime)=
\overline{\cup_{\mathcal{O}_1\subset(\widetilde{\mathcal{O}})^\prime,
\mathcal{O}_1\in\mathcal{K}}\mathcal{A}(\mathcal{O}_1)}^{\Vert\cdot\Vert}$,
where $\widetilde{\mathcal{O}}$ satisfies $\overline{\mathcal{O}}\subsetneq\widetilde{\mathcal{O}}$.
Namely, local states should be regarded as physical processes, and have both functions of observables
and of states as a concept unifying them.
We then consider a local net satisfying `` the split property" and show the existence of local states.

One of the advantages of the use of local states is
to give a clear vision of localization in the context of QFT.
In mathematical physics, Newton and Wigner \cite{NW48}
raised a fundamental question on the difficulty of localization of massless modes.
It is proved by Wightman \cite{Wightman62} that any position operator cannot be defined for
a massless free particle with a non-zero finite spin, especially, for a photon
in the quantum mechanical description.
Two of the authors \cite{OS12} resolved this dilemma by exhibiting the ``effective mass" of a photon due
to the interaction with matter. By contrast, there is room for improving the treatment of localization
in the quantum field theoretical description.
A preparation of a local state is compatible with an actual physical situation (or experimental setting)
and, above all, is much the same as an assignment of a localized space-time region of local excitations.
The dilemma of the absence of ``detection operator"
in a ``localized state" is not an important issue
as far as we use local states since a local state itself separates its localized bounded region from
other causally separated ones.

In section \ref{se:4}, we examine DHR-DR theory in terms of local states. We give a sufficient condition for
minimal Stinespring representations of local states to imply the DHR selection criterion.
This condition mathematically represents physically natural state preparations in vacuum backgrounds.
The discussion in the section gives enough motivation for studying representation theory of CP maps.
We believe that we can improve and use the studies of local aspects of sector theory
developed by Buchholz, Doplicher and Longo \cite{Doplicher,DopLon83,BDL}.

We formulate a theory of orthogonal decomposition of CP maps in section \ref{se:5}. It unifies
a theory of orthogonal decomposition
of states due to Tomita and others (see \cite[Chapter 4]{BR1} and references therein) and
order structure theory of CP maps
by Arveson \cite{Arveson69}, both of which are fundamental and useful in operator algebras.
After the work of Arveson (see \cite{Arveson69,Arveson03,Arveson04,Davidson}
and so on for more details of his work),
Fujimoto \cite{Fujimoto} studied the decomposition theory of CP maps and
established their extremal decompositions.
Our discussion of orthogonal decomposition theory of CP maps follows these studies.
The essence of the theory is summerized in the following theorem
(with some explanations of relevant symbols to be given below):
\setcounter{section}{5}
\setcounter{theorem2}{1}
\begin{theorem2}[Tomita theorem for CP maps]$\quad$\\
For each $T\in \mathrm{CP}(\mathcal{X},\mathcal{M})$,
$\mathcal{O}_T$ and $\mathrm{AbvN}(\pi_T(\mathcal{X})^c)$ are categorically isomorphic.\vspace{1mm}\\
Let $[(S,\mathcal{F},\mu)]\in \mathrm{Ob}(\mathcal{O}_T)$
and $\mathfrak{B}\in\mathrm{Ob}(\mathrm{AbvN}(\pi_T(\mathcal{X})^c))$ be in the categorical isomorphism,
and $(S,\mathcal{F},\mu)$ be a representative of $[(S,\mathcal{F},\mu)]$.
There exists a $^\ast$-isomorphism $\kappa_\mu:L^\infty(\nu)\rightarrow\mathfrak{B}$ defined by
\begin{equation}
V_T^\ast \kappa_\mu(f)\pi_T(X)V_T=\int f(s)\;d\mu(s,X),\hspace{5mm}f\in L^\infty(\nu),X\in\mathcal{X},
\end{equation}
where $\nu$ is a (scalar-valued) positive finite measure which is equivalent to $\rho\circ\mu$
for some normal faithful state $\rho$ on $\mathcal{M}$.
\end{theorem2}
\setcounter{section}{1}
This theorem gives the categorical isomorphism between the category $\mathcal{O}_T$ of equivalence
classes of ``orthogonal CP-measure spaces" which orthogonally decomposes a CP map $T$ and
the category $\mathrm{AbvN}(\pi_T(\mathcal{X})^c)$ of abelian von Neumann subalgebras
of $\pi_T(\mathcal{X})^c$, and enables us to grasp the conceptual meaning of
orthogonal decompositions of CP maps from algebraic viewpoints.

In section \ref{se:6}, generalized local sector theory
is discussed on the basis of mathematics in section \ref{se:5},
which is the sector theory for local states with respect to general reference representations.
This theory contains DHR-DR theory as a special case, and is expected to contribute to
foundations of QFT in a new direction of development of generalized sector theory \cite{Oj03,Oj04,Oj05}.

\section{Preliminaries}\label{se:2}
In this section, we introduce basic notions used in the following sections
and fundamental assumptions in AQFT. The following two subsections can be skipped if you are familiar with
operator algebras and the theory of completely positive maps.

\subsection{Operator algebras and states}\label{subse:2.1}
The contents in this subsection is mainly based on \cite{BR1,T79}.
There are overlaps with \cite{O12,OOS13} but for readers' convenience we have included them.
In this paper, C$^\ast$-algebras $\mathcal{X}$ are assumed to be unital, i.e., $1\in\mathcal{X}$.
We denote by $\mathcal{X}_+$ the set of positive elements of $\mathcal{X}$, i.e.,
$\mathcal{X}_+=\{X\in\mathcal{X}\;|\;X\geq 0\}$.
We denote by $\mathcal{X}^\ast$ the set of linear functionals on $\mathcal{X}$,
and by $\mathcal{X}^\ast_+$ the set of positive linear functionals on $\mathcal{X}$.
We denote by $E_\mathcal{X}$ the set of states on $\mathcal{X}$, i.e.,
$E_\mathcal{X}=\{\omega\in\mathcal{X}^\ast_+\;|\; \omega(1)=1 \}$.
\begin{theorem}[GNS representation theorem \cite{BR1,T79}]
Let $\mathcal{X}$ be a C$^{\ast}$-algebra and $\omega$ be a positive linear functional on $\mathcal{X}$.
Then, there exist a Hilbert space $\mathcal{H}_{\omega}$ with the inner product
$\langle\cdot|\cdot\rangle$, a vector $\Omega_{\omega}\in\mathcal{H}_{\omega}$,
and a $^\ast$-homomorphism $\pi_{\omega}: \mathcal{X}\rightarrow\text{\boldmath $B$}(\mathcal{H}_{\omega})$,
called a $^\ast$-representation (a representation, for short) of $\mathcal{X}$,
such that $\mathcal{H}_\omega=\overline{\pi_\omega(\mathcal{X})\Omega_\omega}$,
where $\overline{\pi_\omega(\mathcal{X})\Omega_\omega}$ denotes
the closure of $\pi_\omega(\mathcal{X})\Omega_\omega$, and
\begin{equation}
\omega(X)=\langle\Omega_{\omega}|\pi_{\omega}(X)\Omega_{\omega}\rangle,\hspace{5mm}X\in\mathcal{X}.
\end{equation}
The triplet $(\pi_{\omega},\mathcal{H}_{\omega}, \Omega_{\omega})$ is called a GNS representation of
$\mathcal{X}$ with respect to $\omega$, and is unique up to unitary equivalence.
\end{theorem}

A linear functional $\rho$ on a von Neumann algebra
$\mathcal{M}$ on $\mathcal{H}$ is said to be normal if there exists
a trace class operator $\sigma$ on $\mathcal{H}$ such that
$\rho(M)=\mathrm{Tr}[\sigma M]$ for each $M\in\mathcal{M}$.
We denote by $\mathcal{M}_\ast$ the set of normal linear functionals on $\mathcal{M}$,
by $\mathcal{M}_{\ast,+}$ the set of normal positive linear functionals on $\mathcal{M}$
and by $\mathcal{M}_{\ast,1}$ the set of normal states on $\mathcal{M}$.
Let $\pi$ be a representation of a C$^\ast$-algebra $\mathcal{X}$. A state $\omega$ on
$\mathcal{X}$ is said to be $\pi$-normal if there exists a normal
state $\rho$ on $\pi(\mathcal{X})^{\prime\prime}$ such that $\omega(X)=\rho(\pi(X))$ for
every $X\in\mathcal{X}$. Two representations $\pi_{1},\pi_{2}$ are
quasi-equivalent, denoted by $\pi_{1}\approx\pi_{2}$, if each $\pi_{1}$-normal
state is $\pi_{2}$-normal and vice versa. As a complement, two representations
$\pi_{1},\pi_{2}$ are disjoint, denoted by 
$\pi_{1}\hspace{0.3em}\raisebox{-.41ex}{$\circ$}\hspace{-0.46em}
\raisebox{0.7ex}{\rotatebox[origin=c]{-90}{--}}\hspace{0.5em}\pi_{2}$, if no
$\pi_{1}$-normal state is $\pi_{2}$-normal and vise versa.

Let $\mathcal{X}$, $\mathcal{Y}$ be C$^\ast$-algebras, and $\mathcal{H}$ be a Hilbert space.
$M_n(\mathcal{X})$ denotes the C$^\ast$-algebra of $n\times n$ matrices with entries $\mathcal{X}$.
Let $T$ be a linear map from $\mathcal{X}$ into $\mathcal{Y}$.
For $n\in\mathbb{N}$, $T^{(n)}:M_n(\mathcal{X})\rightarrow M_n(\mathcal{Y})$ is defined by
$T^{(n)}(X)=(T(X_{ij}))$ for $X=(X_{ij})\in M_n(\mathcal{X})$.
A linear map $T:\mathcal{X}\rightarrow\mathcal{Y}$ is $n$-positive if $T^{(n)}$ is positive.
$T:\mathcal{X}\rightarrow\mathcal{Y}$ is $n$-positive if and only if
for every $X_1,X_2,\cdots, X_n\in\mathcal{X}$, $Y_1,Y_2,\cdots, Y_n\in\mathcal{Y}$,
\begin{equation}
\sum_{i,j=1}^n Y_i^\ast T(X_i^\ast X_j)Y_j\geq 0.
\end{equation}
In the case that $\mathcal{Y}$ is a C$^\ast$-subalgebra of $\textrm{\boldmath $B$}(\mathcal{H})$,
$n$-positivity of $T$ is equivalent to the next condition;
for every $X_1,X_2,\cdots, X_n\in\mathcal{X}$, $x_1,x_2,\cdots, x_n\in\mathcal{H}$,
\begin{equation}
\sum_{i,j=1}^n \langle x_i|T(X_i^\ast X_j) x_j\rangle\geq 0.
\end{equation}
A linear map $T:\mathcal{X}\rightarrow\mathcal{Y}$ is completely positive (CP, for short)
if, for all $n\in\mathbb{N}$, $T^{(n)}$ is positive.
$\mathrm{CP}(\mathcal{X},\mathcal{Y})$ denotes the set of completely positive maps
from $\mathcal{X}$ to $\mathcal{Y}$.
A linear map $T:\mathcal{X}\rightarrow\mathcal{Y}$ is completely bounded (CB, for short)
if $\sup_{n\in\mathbb{N}}\Vert T^{(n)}\Vert<\infty$.
A linear map $T:\mathcal{X}\rightarrow\mathcal{Y}$ is decomposable
if there exist four completely positive maps $T_1,T_2,T_3,T_4:\mathcal{X}\rightarrow\mathcal{Y}$
such that $T=T_1-T_2+i(T_3-T_4)$.

\subsection{Hilbert modules and representation theorems}\label{subse:2.2}

Let $\mathcal{X}$ be a C$^\ast$-algebra and $\mathcal{M}$
be a von Neumann algebra on a Hilbert space $\mathcal{H}$.

\begin{theorem}[Stinespring representation theorem \cite{Stine55,Arveson69,Paulsen02}]\label{Stinespring}
For every $T\in\mathrm{CP}(\mathcal{X},\textrm{\boldmath $B$}(\mathcal{H}))$,
there exist a Hilbert space $\mathcal{K}$, a representation $\pi$ on $\mathcal{K}$
and $V\in\textrm{\boldmath $B$}(\mathcal{H},\mathcal{K})$ such that
\begin{equation}
T(X)=V^\ast \pi(X)V,\hspace{5mm}X\in\mathcal{X}.
\end{equation}
We call the triplet $(\pi,\mathcal{K},V)$
a  Stinespring representation of $T$.
Furthermore, a Stinespring representation $(\pi,\mathcal{K},V)$ of $T$ is said to be minimal if it satisfies
$\mathcal{K}=\overline{\mathrm{span}}(\pi(\mathcal{X})V\mathcal{H})$,
where $\overline{\mathrm{span}}(\pi(\mathcal{X})V\mathcal{H})$ denotes the closure of the space linearly spanned
by $\pi(\mathcal{X})V\mathcal{H}$.
A minimal Stinespring representation of $T$ is denoted by $(\pi_T,\mathcal{K}_T,V_T)$
and is unique up to unitary equivalence.
In the case that $T$ is unital, $V_T$ becomes an isometry.
\end{theorem}
The followng theorem is known as the commutant lifting theorem of Arveson.
\begin{theorem}[Arveson \text{\cite[Theorem 1.3.1]{Arveson69}, \cite[Theorem 12.7]{Paulsen02}}]\label{CommLift}
Let $\mathcal{H}$, $\mathcal{K}$ be Hilbert spaces, $\mathcal{B}$ be a unital C$^\ast$-subalgebra of
$\textrm{\boldmath $B$}(\mathcal{K})$, and $V\in\textrm{\boldmath $B$}(\mathcal{H},\mathcal{K})$ such that
$\mathcal{K}=\overline{\mathrm{span}}(\mathcal{B}V\mathcal{H})$.
For every $A\in(V^\ast \mathcal{B}V)^\prime$, there exists a unique $A_1\in\mathcal{B}^\prime$ such that
$VA=A_1V$. Furthermore, the map $\theta:A\in(V^\ast \mathcal{B}V)^\prime\ni A\mapsto
A_1\in\mathcal{B}^\prime\cap\{VV^\ast\}^\prime$ is a $^\ast$-homomorphism.
\end{theorem}

$E$ is called a Hilbert $\mathcal{M}$-module if it is a right $\mathcal{M}$-module
which has an $\mathcal{M}$-valued inner product $\langle \cdot|\cdot \rangle$ satisfying
\begin{align}
\langle \xi| \eta_1 M_1+\eta_2 M_2\rangle &= \langle \xi| \eta_1\rangle M_1+\langle \xi| \eta_2\rangle M_2,\\
\langle \xi| \eta\rangle^\ast&= \langle \eta| \xi\rangle
\end{align}
for every $\xi,\eta,\eta_1,\eta_2\in E$, $M_1,M_2\in\mathcal{M}$, and is complete with respect to
$\Vert\cdot\Vert_E=\Vert\langle \cdot|\cdot \rangle\Vert^{1/2}$. 
A Hilbert space $\mathcal{H}$ is a typical Hilbert module.
We denote by $\mathcal{B}^a(E)$ the linear space of
adjointable bounded right $\mathcal{M}$-linear operators on $E$
(see \cite{Paschke73,Skeide00} for the definition of the adjointability of operators on a Hilbert module).
$\mathcal{B}^a(E)$ is a C$^\ast$-algebra acting on $E$.
A Hilbert $\mathcal{M}$-module $E$ is said to be self-dual if, for every $\mathcal{M}$-valued
right $\mathcal{M}$-linear functional $f$, there exists $\eta\in E$ such that
$f(\xi)=\langle \eta|\xi \rangle$ for all $\xi\in E$.
For every self-dual Hilbert $\mathcal{M}$-module $E$, $\mathcal{B}^a(E)$ is a W$^\ast$-algebra
\cite[Proposition 3.10]{Paschke73}. The following intrinsic representation theorem then holds:
\begin{theorem}[Paschke \text{\cite[p.465, ll.21--27]{Paschke73}}; Rieffel
\text{\cite[Proposition 6.10]{Rieffel74}; Skeide \cite{Skeide00}}]\label{PGNS}
For every $T\in\mathrm{CP}(\mathcal{X},\mathcal{M})$,
there exist a self-dual Hilbert $\mathcal{M}$-module $E_T$, a $^\ast$-homomorphism $\varpi_T$
from $\mathcal{X}$ into $\mathcal{B}^a(E_T)$,
and $\xi_T\in E_T$ such that $\mathrm{span}(\varpi_T(\mathcal{X})\xi_T\mathcal{M})$
is weak$^\ast$ dense in $E_T$ and that
\begin{equation}
T(X)=\langle \xi_T| \varpi_T(X)\xi_T\rangle,\hspace{5mm}X\in\mathcal{X}.
\end{equation}
We call the triplet $(\varpi_T,E_T,\xi_T)$
a GNS representation of $T$.
A GNS representation of $T$ is unique up to Hilbert module homomorphic equivalence.
\end{theorem}

Let $E$ be a Hilbert $\mathcal{M}$-module,
where $\mathcal{M}$ is a von Neumann algebra on a Hilbert space $\mathcal{H}$.
We then can define a sesquilinear form on $E\otimes \mathcal{H}$ by
\begin{equation}
\left\langle \sum_{i=1}^n\xi_i\otimes x_i\left| \sum_{j=1}^m\eta_j \otimes y_j\right.\right\rangle=
\sum_{i=1}^n\sum_{j=1}^m\langle x_i|\langle \xi_i| \eta_j\rangle y_j\rangle
\end{equation}
for each $\sum_{i=1}^n\xi_i\otimes x_i, \sum_{j=1}^m\eta_j \otimes y_j\in E\otimes \mathcal{H}$.
This sesquilinear form is positive definite. We define a null space
\begin{equation}
\mathcal{N}=\left\{\sum_{i=1}^n\xi_i\otimes x_i\in E\otimes \mathcal{H}\;\left|\;
\left\langle \sum_{i=1}^n\xi_i\otimes x_i\left| \sum_{i=1}^n\xi_i\otimes x_i\right.\right\rangle=0\right.\right\}
\end{equation}
of $E\otimes \mathcal{H}$. The sesquilinear form on $E\otimes \mathcal{H}$
can be extended into that on $E\otimes \mathcal{H}/\mathcal{N}$ which is an inner product on
$E\otimes \mathcal{H}/\mathcal{N}$.
The completion $\overline{E\otimes \mathcal{H}/\mathcal{N}}$ of $E\otimes \mathcal{H}/\mathcal{N}$
with respect to the norm induced by the inner product on $E\otimes \mathcal{H}/\mathcal{N}$
is a Hilbert space and is denoted by $\mathcal{G}_E$.
For every $\xi\in E$, we define a linear map $\lambda(\xi)$ from $\mathcal{H}$ into $\mathcal{G}_E$ by
$\lambda(\xi)x=\xi\otimes x +\mathcal{N}$, $x\in\mathcal{H}$. Then it is easily checked that
$\lambda(\xi)\in\textrm{\boldmath $B$}(\mathcal{H},\mathcal{G}_E)$ and
$\Vert\lambda(\xi)\Vert=\Vert\xi\Vert_E$ for every $\xi\in E$.
The map $\lambda: E\ni\xi \mapsto \lambda(\xi)\in \textrm{\boldmath $B$}(\mathcal{H},\mathcal{G}_E)$
is then a right $\mathcal{M}$-linear isometry.
Furthermore, there exists a faithful representation $\rho:\mathcal{B}^a(E)\rightarrow
\textrm{\boldmath $B$}(\mathcal{G}_E)$ satisfying $\rho(X)\lambda(\xi)=\lambda(X\xi)$
for all $X\in\mathcal{X}$ and $\xi\in E$.
Both $\lambda$ and $\rho$ are weakly continious if $E$ is self-dual.
Since the minimal Stinespring representation of $T$ is unique up to unitary equivalence,
the minimal Stinespring representation of $T$ is unitarily equivalent to the representation
$(\rho\circ\varpi_T,\mathcal{G}_{E_T},\lambda(\xi_T))$ of $T$ \cite[Corollary 7]{Skeide00}.

\subsection{Algebraic quantum field theory}\label{subse:2.3}
We refer readers to \cite{Araki,Haag96} for the details of AQFT.
In addition, see \cite{OOS13}\footnote{In this paper, we give a new axiomatic system of quantum theory
based on sector theory and measurement theory, which establishes the so-called Born rule as a theorem.}
for the details of algebraic quantum theory if they are not familiar with it.

Let $\{\mathcal{A}(\mathcal{O})\}_{\mathcal{O}\in\mathcal{K}}$ be a family of unital $W^\ast$-algebras over
a causal poset $\mathcal{K}$ of bounded subspaces of
the four dimensional Minkowski space $M_4$ satisfying the following conditions:\\
$(1)$ $\mathcal{O}_1\subset \mathcal{O}_2$ $\Rightarrow$
$\mathcal{A}(\mathcal{O}_1)\subset\mathcal{A}(\mathcal{O}_2)$;\\
$(2)$ If $\mathcal{O}_1$ and $\mathcal{O}_2$ are causally separated each other,
then $\mathcal{A}(\mathcal{O}_1)$ and $\mathcal{A}(\mathcal{O}_2)$ mutually commute;\\
$(3)$ The global algebra $\mathcal{A}$ of
$\{\mathcal{A}(\mathcal{O})\}_{\mathcal{O}\in\mathcal{K}}$ is generated as a $C^\ast$-algebra:
\begin{equation}
\mathcal{A}=\overline{\bigcup_{\mathcal{O}\in\mathcal{K}} \mathcal{A}(\mathcal{O})}^{\Vert\cdot\Vert}
\end{equation}
$(4)$ There is a strongly continuous automorphic action $\alpha$ on $\mathcal{A}$
of the Poincare group $\mathcal{P}_+^{\uparrow}$ such that, for any $g=(a,L)\in\mathcal{P}_+^{\uparrow}=
\mathbb{R}^4\rtimes \mathcal{L}_+^{\uparrow}$ and $\mathcal{O}\in\mathcal{K}$,
\begin{equation}
\alpha_{(a,L)}(\mathcal{A}(\mathcal{O}))=\mathcal{A}(L\mathcal{O}+a)=
\mathcal{A}(k_{(a,L)}\mathcal{O}),
\end{equation}
where, for every $g=(a,L)\in\mathcal{P}_+^{\uparrow}$,
$k_g:M_4\rightarrow M_4$ is defined by $k_{(a,L)}x=L x+a$, for all $x\in M_4$.

We call the above $\{\mathcal{A}(\mathcal{O})\}_{\mathcal{O}\in\mathcal{K}}$
a ($W^\ast$-)local net of observables.

In the setting of AQFT,
it is assumed that all physically realizable states on $\mathcal{A}$
and representations of $\mathcal{A}$ are locally normal, i.e.,
normal on $\mathcal{A}(\mathcal{O})$ for all $\mathcal{O}\in\mathcal{K}$.
A vacuum state $\omega_0$ is a $\widetilde{\mathcal{P}_+^{\uparrow}}$-invariant locally normal state
on $\mathcal{A}$ satisfying the following conditions, 
where $\widetilde{\mathcal{P}_+^{\uparrow}}$ is
 the universal covering group of $\mathcal{P}_+^{\uparrow}$:\\
$(1)$ The GNS representation $(\pi_0,\mathcal{H}_0,U,\Omega)$ of $\omega_0$ is an irreducible
representation on a Hilbert space $\mathcal{H}_0$
and a unitary representation $U$ of $\widetilde{\mathcal{P}_+^{\uparrow}}$ such that
\begin{equation}
\pi_0(\alpha_g(X))=U_g\pi_0(X)U_g^\ast,\hspace{5mm}g\in\widetilde{\mathcal{P}_+^{\uparrow}},X\in\mathcal{A};
\end{equation}
$(2)$ $\Omega$ is a cyclic and separating vector for $\pi_0(\mathcal{A}(\mathcal{O}))$,
for any $\mathcal{O}\in\mathcal{K}$;\\
$(3)$ The spectrum of the generator $P=(P_\mu)$ of the translation part of $U$
is contained in the closed future cone $\overline{V_+}$.

$(2)$ is a natural assumption, owing to Reeh-Schlieder theorem.
In addition, we assume that $\mathcal{H}_0$ is separable.

For every $\mathcal{O}\in\mathcal{K}$, we denote by $\overline{\mathcal{O}}$ the closure of $\mathcal{O}$
and define the causal complement $\mathcal{O}^\prime$ of $\mathcal{O}$
by $\mathcal{O}^\prime=\{x\in\mathbb{R}^4\;|\;(x-y)^2<0, y\in\mathcal{O}\}$.
For $\mathcal{O}_1,\mathcal{O}_2\in\mathcal{K}$, we denote by $\mathcal{O}_1\Subset \mathcal{O}_2$
whenever $\overline{\mathcal{O}_1}\subsetneq \mathcal{O}_2$.
Furthermore, we adopt the following notations:
\begin{align}
\mathcal{K}_\Subset=\{\Lambda
&=(\mathcal{O}_1^\Lambda,\mathcal{O}_2^\Lambda)\in\mathcal{K}\times\mathcal{K}\;|\;
\mathcal{O}_1^\Lambda\Subset \mathcal{O}_2^\Lambda\},\\
\mathcal{K}_\Subset^{DC}=\{\Lambda &=(\mathcal{O}_1^\Lambda,\mathcal{O}_2^\Lambda)\in \mathcal{K}_\Subset\;|\;
\mathcal{O}_1^\Lambda\; \text{and}\; \mathcal{O}_2^\Lambda \;\text{are}\;\text{double}\;\text{cones}\}.
\end{align}

$\{\mathcal{A}(\mathcal{O})\}_{\mathcal{O}\in\mathcal{K}}$ has property B
if, for every iclusion pair of regions $\Lambda\in\mathcal{K}_\Subset$
and projection operator $E\in\mathcal{A}(\mathcal{O}_1^\Lambda)$, there is an isometry operator
$W\in\mathcal{A}(\mathcal{O}_2^\Lambda)$ such that $W^\ast W=E$ and $WW^\ast=1$.

We define a dual net $\{\mathcal{A}^d(\mathcal{O})\}_{\mathcal{O}\in\mathcal{K}}$
of $\{\mathcal{A}(\mathcal{O})\}_{\mathcal{O}\in\mathcal{K}}$ with respect to the vacuum representation $\pi_0$
by $\mathcal{A}^d(\mathcal{O}):=\pi_0(\mathcal{A}(\mathcal{O}^\prime))^\prime$
for all double cones $\mathcal{O}\in\mathcal{K}$.
Similarly, we define a bidual net $\{\mathcal{A}^{dd}(\mathcal{O})\}_{\mathcal{O}\in\mathcal{K}}$
of $\{\mathcal{A}(\mathcal{O})\}_{\mathcal{O}\in\mathcal{K}}$ with respect to $\pi_0$ by
$\mathcal{A}^{dd}(\mathcal{O}):=\mathcal{A}^d(\mathcal{O}^\prime)^\prime$,
for all double cones $\mathcal{O}\in\mathcal{K}$, where
$\mathcal{A}^d(\mathcal{O}^\prime)=
\overline{\cup_{\mathcal{O}_1\subset\mathcal{O}^\prime}\mathcal{A}^d(\mathcal{O}_1)}^{\Vert\cdot\Vert}$.
A ($W^\ast$-)local net
$\{\mathcal{A}(\mathcal{O})\}_{\mathcal{O}\in\mathcal{K}}$ satisfies Haag duality in $\pi_0$
if $\mathcal{A}^d(\mathcal{O})=\pi_0(\mathcal{A}(\mathcal{O}))^{\prime\prime}$
for all double cones $\mathcal{O}\in\mathcal{K}$.
$\{\mathcal{A}(\mathcal{O})\}_{\mathcal{O}\in\mathcal{K}}$ satisfies essential duality in $\pi_0$
if $\mathcal{A}^d(\mathcal{O})=\mathcal{A}^{dd}(\mathcal{O})$
for all double cones $\mathcal{O}\in\mathcal{K}$.

Assumptions in AQFT are summerized as follows:
\begin{itemize}
 \item A ($W^\ast$-)local net $\{\mathcal{A}(\mathcal{O})\}_{\mathcal{O}\in\mathcal{K}}$ has
the property B.
 \item $\mathcal{H}_0$ is separable.
 \item The vacuum representation $\pi_0$ is irreducible.
\end{itemize}

\section{Split Property and Local States}\label{se:3}

In algebraic quantum field theory, although the use of the global algebra $\mathcal{A}$
generated by a (W$^\ast$-)local net $\{\mathcal{A}(\mathcal{O})\}_{\mathcal{O}\in\mathcal{K}}$
is an idealization for the purpose of the description of the system,
it is inevitable to define and use the concept of states on it,
which has an attractive potential to evaluate all local observables and to show us its sectors.
We can actually prepare only (normal) states on $\mathcal{A}(\mathcal{O})$ for some $\mathcal{O}\in\mathcal{K}$.
In other words, we do not imagine that we can
specify them on local algebras with infinitely spacelike sparated bounded regions exactly.
In addition to this empirical intuition, 
(normal) states on $\mathcal{A}(\mathcal{O})$ cannot evaluate all elements of $\mathcal{A}$.
Thus, in this section, we find another concept which plays the role of that of (normal) states
on $\mathcal{A}(\mathcal{O})$ for some $\mathcal{O}\in\mathcal{K}$.

First, we pay attention to the following property:
\begin{definition}[]
Let $\{\mathcal{A}(\mathcal{O})\}_{\mathcal{O}\in\mathcal{K}}$ be a local net of observables.
$\{\mathcal{A}(\mathcal{O})\}_{\mathcal{O}\in\mathcal{K}}$ satisfies split property
if, for all $\Lambda\in\mathcal{K}_\Subset$, there exists a type $\mathrm{I}$ factor $\mathcal{N}$ such that
$\mathcal{A}(\mathcal{O}_1^\Lambda)\subset\mathcal{N}\subset\mathcal{A}(\mathcal{O}_2^\Lambda)$.
\end{definition}
The importance of this property was first pointed out by Buchholz \cite{Buchholz74}.
Under several conditions containing the property B and the irreducibility of the vacuum representation $\pi_0$,
it is known that there exist conditions equivalent to the split property in $\pi_0$:
\begin{theorem}[Werner \cite{Werner}; D'Antoni and Longo \cite{DALongo83}]
 $\quad$ \\
The following three conditions are equivalent:\\
$(1)$ $\{\pi_0(\mathcal{A}(\mathcal{O}))^{\prime\prime}\}_{\mathcal{O}\in\mathcal{K}}$ has the split property;\\
$(2)$ For every $\varphi\in\pi_0(\mathcal{A}(\mathcal{O}_1))^{\prime\prime}_{\ast,1}$
and $\mathcal{O}_2\in\mathcal{K}$ such that $\mathcal{O}_1\Subset\mathcal{O}_2$,
there is a unital completely positive map $T$
on $\pi_0(\mathcal{A})^{\prime\prime}=\textrm{\boldmath $B$}(\mathcal{H}_0)$
of the form $T(X)=\sum_j C_j^\ast X C_j$ with $C_j\in\pi_0(\mathcal{A}(\mathcal{O}_2))^{\prime\prime}$
such that $T(X)=\varphi(X)1$, $X\in\pi_0(\mathcal{A}(\mathcal{O}_1))^{\prime\prime}$,
where the summation is $\sigma$-weakly convergent for $\{C_j\}_j$;\\
$(3)$ If $\mathcal{O}_3$ and $\mathcal{O}_4$ are causally separated, then 
\begin{equation}
\pi_0(\mathcal{A}(\mathcal{O}_3))^{\prime\prime}\vee\pi_0(\mathcal{A}(\mathcal{O}_4))^{\prime\prime}
\cong\pi_0(\mathcal{A}(\mathcal{O}_3))^{\prime\prime}\otimes\pi_0(\mathcal{A}(\mathcal{O}_4))^{\prime\prime}.
\end{equation}
\end{theorem}

The most important condition in the paper is (2). This states that every normal states on
$\pi_0(\mathcal{A}(\mathcal{O}_1))^{\prime\prime}$ can be extended into an inner unital CP map
on $\pi_0(\mathcal{A}(\mathcal{O}_2))^{\prime\prime}$ and also on $\pi_0(\mathcal{A})^{\prime\prime}$.
This fact is very natural since an assignment of a state on a bounded space-time region is
a kind of physical operation which is usually described by CP map and called a quantum operation.
For each $\varphi\in\pi_0(\mathcal{A}(\mathcal{O}_1))^{\prime\prime}_{\ast,1}$,
the map $T$ given by the condition (2) is equivalent to the identity map on
$\pi_0(\mathcal{A}(\mathcal{O}_2^\prime))^{\prime\prime}$,
and is called a ``local" quantum operation.
On the basis of the above observation, we define the concept of local states as follows:
\begin{definition}[]\label{LC}
A unital completely positive map $T$ on $\mathcal{A}$ is called a local state on $\mathcal{A}$
with an inclusion pair of regions $\Lambda=(\mathcal{O}_1^\Lambda, \mathcal{O}_2^\Lambda)\in\mathcal{K}_\Subset$
if it satisfies the following conditions:\\
$(1)$  for all $A\in\mathcal{A}$, $B\in\mathcal{A}((\mathcal{O}_2^\Lambda)^\prime):=
\overline{\cup_{\mathcal{O}\in\mathcal{K},\mathcal{O}\subset(\mathcal{O}_2^\Lambda)^\prime}
\mathcal{A}(\mathcal{O})}^{\Vert\cdot\Vert}$, $T(AB)=T(A)B$.\\
$(2)$ there exists $\varphi\in\mathcal{A}(\mathcal{O}_1^\Lambda)_{\ast,1}$
such that $T(A)=\varphi(A)1$ for all $A\in\mathcal{A}(\mathcal{O}_1^\Lambda)$.\\
We denote by $E^L_\mathcal{A}(\Lambda)$ the set of local states on $\mathcal{A}$ with $\Lambda$.
\end{definition}
The concept defined in the above definition plays the role of (normal) state
on $\mathcal{A}(\mathcal{O})$ in a bounded space-time region $\mathcal{O}$, but it is not
a state on the global algebra $\mathcal{A}$. This is the reason why we call it ``local state."
For every local net $\{\mathcal{A}(\mathcal{O})\}_{\mathcal{O}\in\mathcal{K}}$ satisfying
the property B and the split property, and $\Lambda\in\mathcal{K}_{\Subset}$,
we can prove the existence of local states on $\mathcal{A}$ with region $\Lambda$.

Let $T$ be an element of $E^L_\mathcal{A}(\Lambda)$ and $\pi$ be a representation of $\mathcal{A}$.
The map $\pi\circ T$ on $\mathcal{A}$ to $\pi(\mathcal{A})^{\prime\prime}$
is then unital and completely positive:
\begin{equation}
(\pi\circ T)(A)=\pi(T(A)),\hspace{5mm}A\in\mathcal{A}.
\end{equation}
We can apply the representation theorems mentioned in the previous section to
$\pi\circ T\in \mathrm{CP}(\mathcal{A},\pi(\mathcal{A})^{\prime\prime})$.
We can use a restricted definition of local states as a more natural one adapted to the situation:
\begin{definition}[]
A unital CP map $T\in\mathrm{CP}(\mathcal{A},\pi(\mathcal{A})^{\prime\prime})$
is called a local state of $\mathcal{A}$ into $\pi(\mathcal{A})^{\prime\prime}$
with $\Lambda\in\mathcal{K}_\Subset$ if it satisfies the following conditions:\\
$(1)$  for all $A\in\mathcal{A}$ and $B\in\mathcal{A}((\mathcal{O}_2^\Lambda)^\prime)$, $T(AB)=T(A)\pi(B)$.\\
$(2)$ there exists $\varphi\in\mathcal{A}(\mathcal{O}_1^\Lambda)_{\ast,1}$
such that $T(X)=\varphi(X)1$ for all $X\in\mathcal{A}(\mathcal{O}_1^\Lambda)$.\\
We denote by $E^L_{\mathcal{A},\pi}(\Lambda)$
the set of local states of $\mathcal{A}$ into $\pi(\mathcal{A})^{\prime\prime}$ with $\Lambda$.
\end{definition}

\section{Local States and DHR-DR Theory}\label{se:4}
In this section, we assume that the vacuum representation $\pi_0$ satisfies the Haag duality.

\begin{definition}[property DHR \text{\cite[I, pp.228, (A.4)]{DHR2}}]
A representation $\pi$ of $\mathcal{A}$ is said to satisfy property DHR
if, for every $\Lambda\in \mathcal{K}_\Subset^{DC}$ and every projection operator
$E\in \pi^d(\mathcal{O}^\Lambda_1):=\pi(\mathcal{A}((\mathcal{O}^\Lambda_1)^\prime))^\prime$,
there is an isometry $W\in \pi^d(\mathcal{O}^\Lambda_2)$ such that $WW^\ast=E$ and $W^\ast W=1$.
\end{definition}

\begin{lemma}[\text{\cite[I, A.1. Proposition]{DHR2}}]
Let $\omega$ be a state on $\mathcal{A}$ such that
\begin{equation}
\lim_{n\rightarrow\infty}\Vert (\omega-\omega_0)|_{\mathcal{A}(\mathcal{O}_n^\prime)}\Vert=0
\end{equation}
for an increasing sequence of double cones $\{\mathcal{O}_n\}$. If the GNS representation $\pi_\omega$
satisfies the property DHR, there exists a double cone $\mathcal{O}$ such that
\begin{equation}
\pi_\omega|_{\mathcal{A}(\mathcal{O}^\prime)}\cong\pi_0|_{\mathcal{A}(\mathcal{O}^\prime)}.
\end{equation}
This condition is nothing but the DHR selection criterion.
\end{lemma}

\begin{definition}[localized endomorphism]
A $^\ast$-endomorphism $\rho:\pi_0(\mathcal{A})\rightarrow\pi_0(\mathcal{A})$
is called a localized endomorphisms
with support in $\mathcal{O}\in\mathcal{K}$ if it satisfies
$\rho(A)=A$ for all $A\in\pi_0(\mathcal{A}(\mathcal{O}^\prime))$.
\end{definition}
For every representation $\pi$ satisfying the DHR selection criterion for some $\mathcal{O}$,
i.e., $\pi_\omega|_{\mathcal{A}(\mathcal{O}^\prime)}\cong\pi_0|_{\mathcal{A}(\mathcal{O}^\prime)}$,
it is known \cite{DHR2} that there exists a localized endomorphism $\rho$
with support in $\mathcal{O}$ such that $\pi\cong \rho\circ\pi_0$.

For every $T\in E^L_\mathcal{A}(\Lambda)$,
we denote by $(\pi_{T,0},\mathcal{K}_{T,0},V_{T,0})$
the minimal Stinespring representation  of $\pi_0\circ T$.
It is easily checked that
\begin{align}
(\omega_0\circ T)(A) &=\omega_0(T(A))= \langle \Omega|(\pi_0\circ T)(A) \Omega\rangle\nonumber\\
 &= \langle \Omega|V_{T,0}^\ast\pi_{T,0}(A) V_{T,0}\Omega\rangle
 =\langle V_{T,0}\Omega|\pi_{T,0}(A) V_{T,0}\Omega\rangle,\hspace{5mm}A\in\mathcal{A}, \nonumber
\end{align}
and $\Vert (\omega_0\circ T-\omega_0)|_{\mathcal{A}((\mathcal{O}_2^\Lambda)^\prime)}\Vert=0$.
We define a projection $P_{T,0}\in\pi_{T,0}(\mathcal{A})^\prime$
with the range $\mathcal{H}_{\pi_{T,0}}$,
where $\mathcal{H}_{\pi_{T,0}}=\overline{\pi_{T,0}(\mathcal{A}) V_{T,0}\Omega}$.
For any local states $T$ with $\Lambda$, the following two equivalent conditions hold:
\begin{center}
$(i)$ $P_{T,0}=1$, \hspace{5mm}$(ii)$ $\mathcal{K}_{T,0}=\mathcal{H}_{\pi_{T,0}}$.
\end{center}
For the proof,
we use the assumption that $\Omega$ is a cyclic vector for $\pi_0(\mathcal{A}(\mathcal{O}))$,
and the relation $\pi_{T,0}(A)V_{T,0}=V_{T,0}\pi_0(A)$
for all $A\in \mathcal{A}((\mathcal{O}_2^\Lambda)^\prime)$.
Thus we have shown the following theorem:
\begin{theorem}[] \label{SrepDHR}
Let $T$ be a local state on $\mathcal{A}$ with $\Lambda$.
If $\pi_{T,0}$ satisfies the property DHR,
then there is a localized endomorphism $\rho_T$ on $\pi_0(\mathcal{A})$
with region $\mathcal{O}_2^\Lambda$
and $V_{T,0}\in\mathcal{A}^d(\mathcal{O}_2^\Lambda)$
such that
\begin{equation}
(\pi_0\circ T)(X)=V_{T,0}^\ast\rho_T(\pi_0(X))V_{T,0},\hspace{5mm}X\in\mathcal{A}.
\end{equation}
\end{theorem}
This theorem shows that the DHR selection criterion is derived from 
a natural assumption described by the languages of local nets and of local states.

\begin{definition}[transportability]
A localized endomorphism $\rho$ with support in $\mathcal{O}$
is said to be transportable if, for all $a\in\mathbb{R}^4$,
there are a localized endomorphism $\rho^\prime$ with support in $\mathcal{O}+a$ and
a unitary operator $U$ of $\pi_0(\mathcal{A})^{\prime\prime}$ such that $\rho=Ad U\circ\rho^\prime$.
\end{definition}
It is shown in \cite[II, 2.2 Lemma]{DHR1} that two transportable localized endomorphisms
$\rho_1$ and $\rho_2$ are commutative, i.e., $\rho_1\circ\rho_2=\rho_2\circ\rho_1$.
Then the next corollary immediately follows from the above theorem:
\begin{corollary}[]
Let $\Lambda_1,\Lambda_2$ be elements of $\mathcal{K}_\Subset^{DC}$
such that $\mathcal{O}_2^{\Lambda_1}$ and $\mathcal{O}_2^{\Lambda_2}$
are causally separated, and $T_j\in E^L_\mathcal{A}(\Lambda_j)$
satisfying the property DHR, $j=1,2$.
If corresponding localized endomorphisms of $T_1$ and $T_2$ are transportable,
then $\pi_0\circ(T_1\circ T_2)=\pi_0\circ(T_2\circ T_1)$.
\end{corollary}

\begin{proof}
By the definition of $T_1$ and $T_2$, and by Theorem \ref{SrepDHR},
it holds that $V_{T_j,0}\in\mathcal{A}^d(\mathcal{O}_2^{\Lambda_j})$ for each $j=1,2$,
which implies $[V_{T_1},V_{T_2}]=0$.
Thus we have
\begin{align}
(\pi_0\circ(T_1\circ T_2))(A) &=(\pi_0\circ T_1)(T_2(A)) \nonumber\\
 &=V_{T_1,0}^\ast \rho_{T_1}(\pi_0(T_2(A)))V_{T_1,0}\nonumber\\
 &=V_{T_1,0}^\ast \rho_{T_1}((\pi_0\circ T_2)(A))V_{T_1,0} \nonumber\\
 &=V_{T_1,0}^\ast V_{T_2,0}^\ast\rho_{T_1}(\rho_{T_2}(\pi_0(X))) V_{T_2,0}V_{T_1,0} \nonumber\\
 &=V_{T_2,0}^\ast V_{T_1,0}^\ast\rho_{T_2}(\rho_{T_1}(\pi_0(X))) V_{T_1,0}V_{T_2,0} \nonumber\\
 &=V_{T_2,0}^\ast \rho_{T_2}(V_{T_1,0}^\ast\rho_{T_1}(\pi_0(X)) V_{T_1,0})V_{T_2,0} \nonumber\\
 &=V_{T_1,0}^\ast \rho_{T_2}((\pi_0\circ T_1)(A))V_{T_2,0}\nonumber\\
 &= V_{T_2,0}^\ast \rho_{T_2}(\pi_0(T_1(A)))V_{T_2,0} \nonumber\\
 &=(\pi_0\circ T_2)(T_1(A))=(\pi_0\circ(T_1\circ T_2))(A)
\end{align}
for all $A\in\mathcal{A}$.
\end{proof}

\section{Orthogonal CP-measure spaces associated to CP Maps}\label{se:5}

Let $\mathcal{X}$ be a unital C$^\ast$-algebra and
$\mathcal{M}$ be a $\sigma$-finite von Neumann algebra on a Hilbert space $\mathcal{H}$.

\begin{definition}[]
A triplet $(S,\mathcal{F},\mu)$ is called a CP-measure space
if it satisfies the following conditions:\\
$(1)$ $(S,\mathcal{F})$ is a mesurable space,
i.e. $S$ is a set and $\mathcal{F}$ is a $\sigma$-algebra of $S$;\\
$(2)$ $\mu$ is a $\mathrm{CP}(\mathcal{X},\mathcal{M})$-valued measure on $(S,\mathcal{F})$, i.e.,
$\mu:\mathcal{F}\rightarrow \mathrm{CP}(\mathcal{X},\mathcal{M})$
satisfying 
\begin{equation}
\rho(\mu(\cup_i \Delta_i)X)=\sum_i\rho(\mu(\Delta_i)X)
\end{equation}
for all mutually disjoint subset $\{\Delta_i\}_{i\in\mathbb{N}}$ of $\mathcal{F}$,
$\rho\in\mathcal{M}_\ast$ and $X\in\mathcal{X}$.

For every $\mathrm{CP}(\mathcal{X},\mathcal{M})$-valued measure $\mu$ on $(S,\mathcal{F})$,
we also write $\mu(\Delta,X)=\mu(\Delta)X$ for all $\Delta\in\mathcal{F}$ and $X\in\mathcal{X}$.
A CP-measure space $(S,\mathcal{F},\mu)$ is called a CP-measure space with barycenter
$T\in\mathrm{CP}(\mathcal{X},\mathcal{M})$ or a CP-measure space of $T$ if $T=\mu(S)$.
\end{definition}
Let $(S,\mathcal{F},\nu)$ be a measure space and $1\leq p\leq \infty$.
We write $L^p(S,\nu)$ or $L^p(\nu)$ as $L^p(S,\mathcal{F},\nu)$ for short in the paper.
\begin{definition}[]
$(1)$ Let $(S_1,\mathcal{F}_1,\mu_1)$ and $(S_2,\mathcal{F}_2,\mu_2)$ be CP-measure spaces with
barycenter $T$. $(S_1,\mathcal{F}_1,\mu_1)$ is dominated by $(S_2,\mathcal{F}_2,\mu_2)$, 
denoted by $(S_1,\mathcal{F}_1,\mu_1)\prec (S_2,\mathcal{F}_2,\mu_2)$, if
\begin{equation}
\{\mu_1(\Delta_1)\in\mathrm{CP}(\mathcal{X},\mathcal{M})\;|\;\Delta_1\in\mathcal{F}_1\}
\subseteq\{\mu_2(\Delta_2)\in
\mathrm{CP}(\mathcal{X},\mathcal{M})\;|\;\Delta_2\in\mathcal{F}_2\}
\end{equation}
and, for some normal faithful state $\rho$ on ${M}$,
there exists a projection $P\in L^\infty(S_2,\rho\circ \mu_2)$ such that
\begin{equation*}
(L^\infty(S_1,\rho\circ \mu_1),L^2(S_1,\rho\circ \mu_1))
\cong ( PL^\infty(S_2,\rho\circ \mu_2) P,P L^2(S_2,\rho\circ \mu_2)),
\end{equation*}
where for each $j=1,2$ $\rho\circ \mu_j$ is a probability measure on $(S_j,\mathcal{F}_j)$
defined by $(\rho\circ \mu_j)(\Delta)=\rho(\mu_j(\Delta,1))$ for all $\Delta\in\mathcal{F}_j$.\\
$(2)$ $(S_1,\mathcal{F}_1,\mu_1)$ is equivalent to $(S_2,\mathcal{F}_2,\mu_2)$, 
denoted by $(S_1,\mathcal{F}_1,\mu_1)\approx (S_2,\mathcal{F}_2,\mu_2)$, if
$(S_1,\mathcal{F}_1,\mu_1)\prec (S_2,\mathcal{F}_2,\mu_2)$ and
$(S_2,\mathcal{F}_2,\mu_2)\prec (S_1,\mathcal{F}_1,\mu_1)$.
\end{definition}

For any $T_1,T_2\in\mathrm{CP}(\mathcal{X},\mathcal{M})$, we write
$T_1\leq  T_2$ if $T_2-T_1\in \mathrm{CP}(\mathcal{X},\mathcal{M})$ \cite{Arveson69}.
It is  obvious that $\leq$ is an order relation on $\mathrm{CP}(\mathcal{X},\mathcal{M})$.
The order $\leq$ on $\mathrm{CP}(\mathcal{X},\mathcal{M})$ can be interpreted as
the restriction of the order $\leq^\prime$ on $\mathrm{CP}(\mathcal{X},\textrm{\boldmath $B$}(\mathcal{H}))$
into elements of $\mathrm{CP}(\mathcal{X},\textrm{\boldmath $B$}(\mathcal{H}))$ with range $\mathcal{M}$.
We do not distinguish $\leq$ and $\leq^\prime$ from now on.
Then the following lemma and theorem are known to hold:
\begin{lemma}[Arveson \text{\cite[Lemma 1.4.1]{Arveson69}}]\label{Arveson1}
Let $T$ and $T^\prime$ be elements of $\mathrm{CP}(\mathcal{X},\textrm{\boldmath $B$}(\mathcal{H}))$
such that $T^\prime\leq  T$. There exists $R\in \pi_T(\mathcal{X})^\prime$ such that $0\leq R\leq 1$ and
\begin{equation}
T^\prime(X)= V_T^\ast R\pi_T(X)V_T,\hspace{5mm}X\in\mathcal{X}.
\end{equation}
\end{lemma}

\begin{theorem}[Arveson \text{\cite[Theorem 1.4.2]{Arveson69}}]\label{Arveson2}
Let $T$ be an element of $\mathrm{CP}(\mathcal{X},\textrm{\boldmath $B$}(\mathcal{H}))$ and
$(\pi_T,\mathcal{K}_T,V_T)$ be the minimal Stinespring representation of $T$.
There exists an affine order isomorphism between
$[0,T]=\{T^\prime\in\mathrm{CP}(\mathcal{X},\textrm{\boldmath $B$}(\mathcal{H}))\;|\;0\leq T^\prime\leq T\}$
and $\{R\in\pi_T(\mathcal{X})^\prime\;|\; 0\leq R\leq 1\}$.
\end{theorem}

Since $\mathcal{M}\subset\textrm{\boldmath $B$}(\mathcal{H})$ implies
$\mathrm{CP}(\mathcal{X},\mathcal{M})\subset\mathrm{CP}(\mathcal{X},\textrm{\boldmath $B$}(\mathcal{H}))$,
for every $T\in\mathrm{CP}(\mathcal{X},\mathcal{M})$
there exists a minimal Stinespring representation $(\pi_T,\mathcal{K}_T,V_T)$ of $T$.
Then, for a $R\in\{R\in\pi_T(\mathcal{X})^\prime\;|\; 0\leq R\leq 1\}$,
$T_R(X):=V_T^\ast R\pi_T(X)V_T$, $X\in\mathcal{X}$, is
an element of $\mathrm{CP}(\mathcal{X},\textrm{\boldmath $B$}(\mathcal{H}))$ but 
it is not always that of $\mathrm{CP}(\mathcal{X},\mathcal{M})$. We consider a set
\begin{equation}
\{R\in\pi_T(\mathcal{X})^\prime\;|\; 0\leq R\leq 1,T_R\in\mathrm{CP}(\mathcal{X},\mathcal{M})\}.
\end{equation}
This set is seen to be affinely order isomorphic to
$\{T^\prime\in\mathrm{CP}(\mathcal{X},\mathcal{M})\;|\;0\leq T^\prime\leq T\}$.
We denote by $\pi_T(\mathcal{X})^c$ the $\sigma$-weak closure of
a linear subspace of $\textrm{\boldmath $B$}(\mathcal{H})$ linearly spanned by
$\{R\in\pi_T(\mathcal{X})^\prime\;|\; 0\leq R\leq 1,T_R\in\mathrm{CP}(\mathcal{X},\mathcal{M})\}$
and call this operator system an $\mathcal{M}$-relative commutant of $\pi_T(\mathcal{X})$.
\begin{lemma}\label{M-rel}
$\pi_T(\mathcal{X})^c$ is a von Neumann algebra on $\mathcal{K}_T$.
\end{lemma}
\begin{proof}
By \cite[Proposition 5.4]{Paschke73}, there exists an affine order isomorphism between
$[0,T]=\{T^\prime\in\mathrm{CP}(\mathcal{X},\mathcal{M})\;|\;0\leq T^\prime\leq T\}$
and $\{R\in\varpi_T(\mathcal{X})^\prime\;|\; 0\leq R\leq 1\}$,
where $\varpi_T(\mathcal{X})^\prime=\{A\in\mathcal{B}^a(E_T)\;|\;AB=BA\;
\text{for}\;\text{all}\;B\in \varpi_T(\mathcal{X})\}$.
Let $\rho:\mathcal{B}^a(E_T)\rightarrow\textrm{\boldmath $B$}(\mathcal{G}_{E_T})$
be a normal faithful representation (see subsection \ref{subse:2.2}).
By the unitary equivalence of $\rho\circ\varpi_T$ and $\pi_T$,
there exists a unitary operator $U: \mathcal{G}_{E_T}\rightarrow \mathcal{K}_T$
such that $\pi_T(X)=U(\rho\circ\varpi_T)(X)U^\ast$ for all $X\in\mathcal{X}$.
Then there is a one-to-one correspondence between
$\{URU^\ast|\;R\in\rho(\varpi_T(\mathcal{X})^\prime),0\leq R\leq 1\}$
and $\{R\in\pi_T(\mathcal{X})^c\;|\; 0\leq R\leq 1\}$,
which completes the proof.
\end{proof}

\begin{lemma}[]\label{StineOrth}
Let $T_1$ and $T_2$ be elements of $\mathrm{CP}(\mathcal{X},\mathcal{M})$,
and $T=T_1+T_2$.
The following three conditions are equivalent:\\
$(1)$ $(\pi_{T},\mathcal{K}_{T},V_{T})=(\pi_{T_1},\mathcal{K}_{T_1},V_{T_1})
\oplus (\pi_{T_2},\mathcal{K}_{T_2},V_{T_2})$G\\
$(2)$ There exists a projection $P\in\pi_T(\mathcal{X})^c$ such that
\begin{equation}
T_1(X)= V_T^\ast P\pi_T(X)V_T,\hspace{3mm} T_2(X)=V_T^\ast (1-P)\pi_T(X)V_T,\hspace{5mm}X\in\mathcal{X};
\end{equation}
$(3)$ If $T^\prime\in\mathrm{CP}(\mathcal{X},\mathcal{M})$ satisfies
$T^\prime \leq T_1$ and $T^\prime \leq T_2$, then $T^\prime=0$.\\
If $T_1$ and $T_2$ satisfy the above equivalent conditions,
they are said to be (mutually) orthogonal, denoted by $T_1\;\bot\; T_2$.
\end{lemma}
\begin{proof}
The proof of the lemma is similar to that of \cite[Lemma 4.1.19]{BR1}.\\
$(1)\Rightarrow (2)$ Let $P$ be a projection on $\mathcal{K}_T$ with the range $\mathcal{K}_{T_1}$.
It is obvious that $P\in\pi_T(\mathcal{X})^c$, and that 
$T_1(X)= V_T^\ast P\pi_T(X)V_T$, $T_2(X)=V_T^\ast (1-P)\pi_T(X)V_T$, $X\in\mathcal{X}$.\\
$(2)\Rightarrow (1)$ If one sets $(\pi_{1},\mathcal{K}_{1},V_{1})=(P\pi_{T},P\mathcal{K}_{T},PV_{T})$,
then it follows from the uniqueness statement of Theorem
\ref{Stinespring} that $(\pi_{1},\mathcal{K}_{1},V_{1})$
is unitarily equivalent to $(\pi_{T_1},\mathcal{K}_{T_1},V_{T_1})$.
Similarly, $((1-P)\pi_{T},(1-P)\mathcal{K}_{T},(1-P)V_{T})$ is unitarily equivalent to
$(\pi_{T_2},\mathcal{K}_{T_2},V_{T_2})$.\\
$(2)\Rightarrow (3)$ Suppose that $T^\prime\in\mathrm{CP}(\mathcal{X},\mathcal{M})$ satisfies
$T^\prime \leq T_1$ and $T^\prime \leq T_2$.
By Lemma \ref{Arveson1}, there exists a positive operator $R\in \pi_T(\mathcal{X})^c$ such that
\begin{equation}
T^\prime(X)= V_T^\ast R\pi_T(X)V_T,\hspace{5mm}X\in\mathcal{X},
\end{equation}
and that $R\leq P$ and $R\leq 1-P$. Then, one has $0\leq (1-P)R(1-P)\leq (1-P)P(1-P)=0$ and
$0\leq PRP\leq P(1-P)P=0$. Hence, $R^{1/2}(1-P)=R^{1/2}P=0$, and so $R^{1/2}=0$.
Thus $T^\prime=0$.\\
$(3)\Rightarrow (2)$
By Lemma \ref{Arveson1}, there exists a positive operator $R\in \pi_T(\mathcal{X})^c$ such that
$0\leq R\leq 1$ and
\begin{equation}
T_1(X)= V_T^\ast R\pi_T(X)V_T,\hspace{5mm}X\in\mathcal{X}.
\end{equation}
If one introduces $R^\prime=R(1-R)$,
then $R^\prime\in (\pi_T(\mathcal{X})^c)_+$ by Lemma \ref{M-rel}
and the linear map $T^\prime$ defined by
\begin{equation}
T^\prime(X)= V_T^\ast R^\prime\pi_T(X)V_T,\hspace{5mm}X\in\mathcal{X},
\end{equation}
is completely positive. Then $T^\prime\leq T_1$, $T^\prime\leq T_2$, and condition $(3)$
together with cyclicity of $V_T$ imply that $R(1-R)=0$, i.e., $R$ is a projection operator on $\mathcal{K}_T$.
\end{proof}
\begin{definition}[]
A triplet $(S,\mathcal{F},\mu)$ is said to be an orthogonal CP-measure space with barycenter
$T\in\mathrm{CP}(\mathcal{X},\mathcal{M})$ if it is a CP-measure space of $T$ and satisfies
$\mu(\Delta)$ $\bot$ $\mu(\Delta^c)$ in the sense of Lemma \ref{StineOrth},
for all $\Delta\in\mathcal{F}$.
\end{definition}

\begin{lemma}[]\label{tomita5}
If $(S,\mathcal{F},\mu)$ is a CP-measure space of $T\in\mathrm{CP}(\mathcal{X},\mathcal{M})$,
then there exists a unique linear map $L^\infty(\nu)\ni f\mapsto \kappa_\mu(f)\in\pi_T(\mathcal{X})^c$
defined by
\begin{equation}
V_T^\ast \kappa_\mu(f)\pi_T(X)V_T=\int f(s)\;d\mu(s,X),\hspace{5mm}f\in L^\infty(\nu),X\in\mathcal{X},
\end{equation}
i.e., for every $\rho\in\mathcal{M}_\ast$,
\begin{equation}
\rho(V_T^\ast \kappa_\mu(f)\pi_T(X)V_T)=\int f(s)\;d(\rho(\mu(s,X))),
\hspace{5mm}f\in L^\infty(\nu),X\in\mathcal{X}, \label{tomita0}
\end{equation}
which is positive and contractive,
where $\nu$ is a (scalar-valued) positive finite measure which is equivalent to $\rho\circ\mu$
for some normal faithful state $\rho$ on $\mathcal{M}$.
Furthermore, if $f\in L^\infty(\nu)_+$ satisfies $\kappa_\mu(f)=0$, then $f=0$.\\
If $L^\infty(\nu)$ is equipped with the $\sigma(L^\infty(\nu),L^1(\nu))$-topology and
$\pi_T(\mathcal{X})^c$ with the weak topology, then the map $\kappa_\mu$ is continuous.
\end{lemma}
\begin{remark}
We can easily check the well-definedness of the right-hand side of Eq. (\ref{tomita0}).
To begin with, $\rho\circ\mu\ll\nu$ holds for all $\rho\in\mathcal{M}_{\ast,+}$.
For every $X\in\mathcal{X}$ and $\rho\in\mathcal{M}_{\ast}$, we define
a finite measure $\langle \rho\circ\mu, X\rangle$ on $(S,\mathcal{F})$ by
$\langle \rho\circ\mu, X\rangle(\Delta)=\rho(\mu(\Delta,X))$ for all $\Delta\in\mathcal{F}$.
For every $\Delta\in\mathcal{F}$, $ X\in\mathcal{X}_+$ and $\rho\in\mathcal{M}_{\ast,+}$,
\begin{equation}
\langle \rho\circ\mu, X\rangle(\Delta)=\rho(\mu(\Delta,X))\leq \Vert X\Vert\cdot(\rho\circ\mu)(\Delta).
\end{equation}
Thus $\langle \rho\circ\mu, X\rangle\ll\nu$
for all $X\in\mathcal{X}_+$ and $\rho\in\mathcal{M}_{\ast,+}$.
Since every element of $\mathcal{X}$ is a combination of four positive elements of $\mathcal{X}$
and $\mathcal{M}_\ast$ is linearly spanned by $\mathcal{M}_{\ast,+}$,
$\langle \rho\circ\mu, X\rangle\ll\nu$
holds for all $X\in\mathcal{X}$ and $\rho\in\mathcal{M}_{\ast}$.
\end{remark}
\begin{proof}[Proof of Lemma \ref{tomita5}]
For every $f\in L^\infty(\nu)$, we define a linear map $T_f$
on $\mathcal{X}$ into $\mathcal{M}$ by
\begin{equation}
\rho(T_f(X))=\int f(s)\;d(\rho(\mu(s,X))),
\hspace{5mm}\rho\in\mathcal{M}_\ast,X\in\mathcal{X}.
\end{equation}
For every $f\in L^\infty(\nu)$, $T_f$ is decomposable
since $T_{f^\prime}$ is completely positive for every $f^\prime\in L^\infty(\nu)_+$
and $T_f$ has the form $T_f=T_{f_1}-T_{f_2}+i(T_{f_3}-T_{f_4})$, where
$f_1,f_2,f_3,f_4\in L^\infty(\nu)_+$ such that $f=f_1-f_2+i(f_3-f_4)$.
For every $X_1,X_2,\cdots,X_n\in\mathcal{X}$ and $\xi_1,\xi_2,\cdots,\xi_n\in\mathcal{H}$,
\begin{align}\label{cp}
\left|\sum_{i,j=1}^n\langle \xi_i| T_f(X_i^\ast X_j)\xi_j\rangle \right|
&=\left| \int f(s)\;
 d\left(\sum_{i,j=1}^n\langle \xi_i| \mu(s,X_i^\ast X_j)\xi_j\rangle\right) \right| \nonumber \\
 &\leq \Vert f\Vert_\infty \left|\sum_{i,j=1}^n \int
 d\langle \xi_i| \mu(s,X_i^\ast X_j)\xi_j\rangle \right|  \nonumber \\
 &= \Vert f\Vert_\infty \sum_{i,j=1}^n\langle \xi_i| T(X_i^\ast X_j)\xi_j\rangle. 
\end{align}
Especially, if $f\geq 0$, then
\begin{equation}
\sum_{i,j=1}^n\langle \xi_i| T_f(X_i^\ast X_j)\xi_j\rangle \leq
\Vert f\Vert_\infty \sum_{i,j=1}^n\langle \xi_i| T(X_i^\ast X_j)\xi_j\rangle,
\end{equation}
i.e., $T_f\leq \Vert f\Vert_\infty \cdot T$. Then there exists
$\kappa_\mu(f)\in\pi_T(\mathcal{X})^c$ such that
\begin{equation}
T_f(X)=V_T^\ast\kappa_\mu(f)\pi_T(X)V_T,\hspace{5mm}X\in\mathcal{X},
\end{equation}
and that $\Vert\kappa_\mu(f)\Vert\leq \Vert f\Vert_\infty$.
Since every $f\in L^\infty(\nu)$ has the form $f=f_1-f_2+i(f_3-f_4)$, $f_1,f_2,f_3,f_4\in L^\infty(\nu)_+$,
$\kappa_\mu(f)$ can be defined by $\kappa_\mu(f)=\kappa_\mu(f_1)-\kappa_\mu(f_2)+
i(\kappa_\mu(f_3)-\kappa_\mu(f_4))$. The estimate
$\Vert\kappa_\mu(f)\Vert\leq \Vert f\Vert_\infty$, for every $f\in L^\infty(\nu)$,
follows from Eq. (\ref{cp}).
Let $f$ be an element of $L^\infty(\nu)_+$ such that $\kappa_\mu(f)=0$.
For a normal faithful state $\varphi$ on $\mathcal{M}$, it holds that
\begin{equation}
0=\varphi(V_T^\ast\kappa_\mu(f)V_T)=\int f(s)\;d(\varphi\circ\mu)(s)=(\varphi\circ\mu)(f).
\end{equation}
Since $\varphi\circ\mu$ is a normal faithful (semifinite) trace on $L^\infty(\nu)$
and $f\in L^\infty(\nu)_+$, $f=0$.

For every $m,n\in\mathbb{N}$, $X_1,X_2,\cdots,X_m$, $Y_1,Y_2,\cdots,Y_n$ $\in\mathcal{X}$,
and $\xi_1,\xi_2,\cdots,\xi_m$, $\eta_1,\eta_2,\cdots,\eta_n$ $\in\mathcal{H}$, it holds that
\begin{equation}
\left\langle \left. \sum_{i=1}^m\pi_T(X_i)V_T\xi_i
\right|\kappa_\mu(f)\left(\sum_{j=1}^n\pi_T(Y_j)V_T\eta_j\right)\right\rangle
=\sum_{i=1}^m\sum_{j=1}^n\int f(s)\;d\langle \xi_i | \mu(s,X_i^\ast Y_j)\eta_j\rangle.
\end{equation}
Since $\sum_{i=1}^m\sum_{j=1}^n\langle \xi_i | \mu(\cdot,X_i^\ast Y_j)\eta_j\rangle\ll \nu$,
the functional $L^\infty(\nu)\ni f\mapsto$ $\langle \sum_{i=1}^m\pi_T(X_i)V_T\xi_i$
$|\kappa_\mu(f)(\sum_{j=1}^n \pi_T(Y_j)V_T\eta_j)\rangle$ $\in\mathbb{C}$ is continuous.
The continuity of $\kappa_\mu$ follows from the boundedness of $\kappa_\mu(f)$
for every $f\in L^\infty(\nu)$
and from the density of $\mathrm{span}(\pi_T(\mathcal{X})V_T\mathcal{H})$ in $\mathcal{K}_T$.
\end{proof}
\begin{lemma}[]\label{tomita6}
Let $(S,\mathcal{F},\mu)$ be a CP-measure space space of $T\in\mathrm{CP}(\mathcal{X},\mathcal{M})$.
The following three conditions are equivalent:\\
$(1)$ $(S,\mathcal{F},\mu)$ is an orthogonal CP-measure space of $T$;\\
$(2)$ the map $L^\infty(\nu)\ni f\mapsto \kappa_\mu(f)\in\pi_T(\mathcal{X})^c$ is a $^\ast$-isomorphism
of $L^\infty(\nu)$ into $\pi_T(\mathcal{X})^c$;\\
$(3)$ the map $L^\infty(\nu)\ni f\mapsto \kappa_\mu(f)\in\pi_T(\mathcal{X})^c$ is a $^\ast$-representation.\\
If the above equivalent conditions are satisfied, then
$\mathfrak{B}=\{\kappa_\mu(f)\in\pi_T(\mathcal{X})^c\;|\;f\in L^\infty(\nu)\}$ is
an abelian von Neumann subalgebra of $\pi_T(\mathcal{X})^c$.
\end{lemma}
\begin{proof}
$(1)\Rightarrow (2)$ Assume $(1)$. $\kappa_\mu$ is a positive linear map from $L^\infty(\nu)$
into $\pi_T(\mathcal{X})^c$ by the previous lemma. For a projection $f$ of $L^\infty(\nu)$,
there is $\Delta\in\mathcal{F}$ such that $f=\chi_\Delta$. By assumption, it holds that
\begin{equation}
\mu(\Delta)\hspace{2mm}\bot\hspace{2mm}\mu(\Delta^c).
\end{equation}
On the other hand, it holds that
\begin{equation}\label{hom}
\mu(\Delta)+\mu(\Delta^c)=T.
\end{equation}
Hence $\kappa_\mu(f)$ is a projection by Lemma \ref{tomita5}. 
For mutually orthogonal projections $f$ and $g$ of $L^\infty(\nu)$, $f\leq 1-g$.
By Lemma \ref{Arveson1}, it holds that $\kappa_\mu(f)\leq \kappa_\mu(1-g)=1-\kappa_\mu(g)$.
Hence, it holds that $\kappa_\mu(f)\kappa_\mu(g)=0$. Thus, for all projections $f,g$ of $L^\infty(\nu)$,
\begin{equation}
\kappa_\mu(fg)=\kappa_\mu(f)\kappa_\mu(g).
\end{equation}
Since every element of $L^\infty(\nu)$ can be approximated by simple functions on $S$
in norm and the estimate $\Vert\kappa_\mu(f)\Vert\leq \Vert f\Vert_\infty$ holds
for all $f\in L^\infty(\nu)$, the equality (\ref{hom}) extends to all $f,g\in L^\infty(\nu)$.
This implies that $\kappa_\mu$ is a $^\ast$-homomorphism.
Also, for $f\in L^\infty(\nu)$ and $\rho\in\mathcal{M}_{\ast,+}$,
\begin{equation}
\rho((\kappa_\mu(f)V_T)^\ast\kappa_\mu(f)V_T)=\int |f(s)|^2\;d(\rho\circ\mu)(s).
\end{equation}
By this relation, $\kappa_\mu$ is faithful.\\
$(2)\Rightarrow (3)$ A $^\ast$-isomorphism into $\pi_T(\mathcal{X})^c$ is obviously a $^\ast$-representation.\\
$(3)\Rightarrow (1)$ Assume $(3)$. By assumption, for each $\Delta\in\mathcal{F}$,
$\kappa_\mu(\chi_\Delta)$ and $\kappa_\mu(\chi_{\Delta^c})$ are mutually orthogonal projections
satisfying $\kappa_\mu(\chi_\Delta)+\kappa_\mu(\chi_{\Delta^c})=1$. Therefore, it is seen, by
Lemma \ref{tomita5}, that
\begin{equation}
\mu(\Delta)\hspace{2mm}\bot\hspace{2mm}\mu(\Delta^c).
\end{equation}
See \cite[Proposition 4.1.22]{BR1} for the proof of the last part of the lemma.
\end{proof}
For any category $E$, we denote by $\mathrm{Ob}(E)$ and $\mathrm{Arrow}(E)$
the objects and the arrows of $E$, respectively.
We denote by $\mathcal{O}_T$ the category of $\approx$-equivalence classes of
orthogonal CP-measure spaces with barycenter $T\in\mathrm{CP}(\mathcal{X},\mathcal{M})$,
whose arrows are defined by the dominance relation between representatives of equivalence classes,
i.e., for $[(S_1,\mathcal{F}_1,\mu_1)],[(S_2,\mathcal{F}_2,\mu_2)]\in \mathrm{Ob}(\mathcal{O}_T)$,
$[(S_1,\mathcal{F}_1,\mu_1)]\rightarrow[(S_2,\mathcal{F}_2,\mu_2)]$
$\in \mathrm{Arrow}(\mathcal{O}_T)$ if
$(S_1,\mathcal{F}_1,\mu_1)$ $\prec(S_2,\mathcal{F}_2,\mu_2)$, where
$(S_j,\mathcal{F}_j,\mu_j)$ $\in[(S_j,\mathcal{F}_j,\mu_i)]$, $j=1,2$.
Also, we denote by $\mathrm{AbvN}(\pi_T(\mathcal{X})^c)$ the category of abelian von Neumann
subalgebras of $\pi_T(\mathcal{X})^c$,
whose arrows are defined by the inclusion of
von Neumann algebras. The following theorem is the main result in this section:
\begin{theorem}[Tomita theorem for CP maps]
For each $T\in \mathrm{CP}(\mathcal{X},\mathcal{M})$,
$\mathcal{O}_T$ and $\mathrm{AbvN}(\pi_T(\mathcal{X})^c)$ are categorically isomorphic.\vspace{1mm}\\
Let $[(S,\mathcal{F},\mu)]\in \mathrm{Ob}(\mathcal{O}_T)$
and $\mathfrak{B}\in\mathrm{Ob}(\mathrm{AbvN}(\pi_T(\mathcal{X})^c))$ be in the categorical isomorphism,
and $(S,\mathcal{F},\mu)$ be a representative of $[(S,\mathcal{F},\mu)]$.
There exists a $^\ast$-isomorphism $\kappa_\mu:L^\infty(\nu)\rightarrow\mathfrak{B}$ defined by
\begin{equation}
V_T^\ast \kappa_\mu(f)\pi_T(X)V_T=\int f(s)\;d\mu(s,X),\hspace{5mm}f\in L^\infty(\nu),X\in\mathcal{X},
\end{equation}
where $\nu$ is a (scalar-valued) positive finite measure which is equivalent to $\rho\circ\mu$
for some normal faithful state $\rho$ on $\mathcal{M}$.
\end{theorem}
\begin{proof}
Let $[(S,\mathcal{F},\mu)]\in \mathrm{Ob}(\mathcal{O}_T)$
and $(S,\mathcal{F},\mu)$ be a representative of $[(S,\mathcal{F},\mu)]$.
By Lemma \ref{tomita6}, there exists a $^\ast$-isomorphism
$\kappa_\mu:L^\infty(\nu)\ni f\mapsto \kappa_\mu(f)\in\pi_T(\mathcal{X})^c$
of $L^\infty(\nu)$ into $\pi_T(\mathcal{X})^c$,
where $\nu$ is a (scalar-valued) positive finite measure which is equivalent to $\rho\circ\mu$
for some normal faithful state $\rho$ on $\mathcal{M}$, and the range
$\mathfrak{B}_{[(S,\mathcal{F},\mu)]}:=\kappa_\mu(L^\infty(\nu))$ is an abelian von Neumann
subalgebra of $\pi_T(\mathcal{X})^c$, i.e.,
$\mathfrak{B}_{[(S,\mathcal{F},\mu)]}\in \mathrm{Ob}(\mathrm{AbvN}(\pi_T(\mathcal{X})^c))$.
It is also seen that $\mathfrak{B}_{[(S,\mathcal{F},\mu)]}$
does not depend on the choice of a representative $(S,\mathcal{F},\mu)$ of $[(S,\mathcal{F},\mu)]$.

Let $\mathfrak{B}\in\mathrm{Ob}(\mathrm{AbvN}(\pi_T(\mathcal{X})^c))$.
By \cite[Chapter III, Theorem 1.18]{T79}, there exist a locally compact Hausdorff space $\Gamma$,
a positive regular Borel measure $\lambda$ on $\Gamma$, and a $^\ast$-isomorphism
$\alpha: L^\infty(\Gamma,\lambda)\rightarrow \mathfrak{B}$. Then we define
$\mathrm{CP}(\mathcal{X},\mathcal{M})$-valued measure $\mu$ by
\begin{equation}
\mu_\mathfrak{B}(\Delta,X)=V_T^\ast\pi_T(X)\alpha(\chi_\Delta)V_T,
\end{equation}
for all $\Delta\in\mathcal{F}$ and $X\in\mathcal{X}$.
By Lemma \ref{tomita6}, $(\Gamma,\mathcal{B}(\Gamma),\mu_\mathfrak{B})$ is
an orthogonal CP-measure space with barycenter $T$.

By the preceeding two paragraphs, it is shown that there exists a one-to-one correspondence
between $\mathrm{Ob}(\mathcal{O}_T)$ and $\mathrm{Ob}(\mathrm{AbvN}(\pi_T(\mathcal{X})^c))$.
This correspondence can be extended into an isomorphism between
$\mathcal{O}_T$ and $\mathrm{AbvN}(\pi_T(\mathcal{X})^c)$ naturally.
\end{proof}
\begin{definition}[]
Let $T_1,T_2\in\mathrm{CP}(\mathcal{X},\mathcal{M})$.\\
$(1)$ $T_1$ and $T_2$ are said to be quasi-equivalent, denoted by $T_1\approx T_2$,
if so are the minimal Stinespring representations $\pi_{T_1}$ and $\pi_{T_2}$.\\
$(2)$ $T_1$ and $T_2$ are said to be disjoint, denoted by $T_1\hspace{0.3em}
\raisebox{-.41ex}{$\circ$}\hspace{-0.53em}\raisebox{0.7ex}{\rotatebox[origin=c]{-90}{--}}\hspace{0.5em}T_{2}$,
if so are $\pi_{T_1}$ and $\pi_{T_2}$.
\end{definition}

\begin{lemma}[]\label{Sdisjoint}
Let $T_1,T_2\in\mathrm{CP}(\mathcal{X},\mathcal{M})$ and $T=T_1+T_2$.
The following conditions are equivalent:\\
$(1)$ $T_1\hspace{0.3em}
\raisebox{-.41ex}{$\circ$}\hspace{-0.53em}
\raisebox{0.7ex}{\rotatebox[origin=c]{-90}{--}}\hspace{0.5em}T_2$G\\
$(2)$ There exists a projection $P\in\pi_T(\mathcal{X})^{\prime\prime}\cap
\pi_T(\mathcal{X})^c$ such that
\begin{equation*}
T_1(X) = V_T^\ast P\pi_T(X)V_T,\hspace{1mm}
T_2(X) = V_T^\ast (1-P)\pi_T(X)V_T,\hspace{3mm}X\in\mathcal{X}.
\end{equation*}
\end{lemma}
\begin{proof} $\quad$\\
$(1)\Rightarrow (2)$ Assume $(1)$.
Let $T^\prime\in\mathrm{CP}(\mathcal{X},\mathcal{M})$ such that $T^\prime\leq T_1$ and
$T^\prime\leq T_2$. 
By Lemma \ref{Arveson1}, $T^\prime$ has the form
\begin{equation}
T^\prime(X)=V_{T_1}^\ast Q_1\pi_{T_1}(X)V_{T_1}=V_{T_2}^\ast Q_2\pi_{T_2}(X)V_{T_2},
\end{equation}
for all $X\in\mathcal{X}$, where $Q_j\in\pi_{T_j}(\mathcal{X})^c$, $j=1,2$.
Thus the minimal Stinespring representation of $T^\prime$ is unitarily equivalent to
subrepresentations of $\pi_{T_1}$ and of $\pi_{T_2}$.
By assumption, which is nothing but the disjointness of $\pi_{T_1}$ and $\pi_{T_2}$,
we have $T^\prime=0$. By Lemma \ref{StineOrth}, $T_1\;\bot\; T_2$.
Therefore, by Lemma \ref{Arveson1}, there exists a projection $P\in\pi_T(\mathcal{X})^\prime$ such that
\begin{equation}
T_1(X)=V_T^\ast P\pi_T(X)V_T,\hspace{5mm}T_2(X)=V_T^\ast (1-P)\pi_T(X)V_T,
\end{equation}
for all $X\in\mathcal{X}$. In particular, $P\in\pi_T(\mathcal{X})^c$.
For any $B\in\pi_T(\mathcal{X})^\prime$ and $W\in\textrm{\boldmath $B$}(\mathcal{H},P\mathcal{K}_T)$,
we define $T^{\prime\prime}\in\mathrm{CP}(\mathcal{X},\textrm{\boldmath $B$}(\mathcal{H}))$ by
\begin{equation}
T^{\prime\prime}(X)=((1-P)BW)^\ast\pi_T(X)((1-P)BW),
\end{equation}
for all $X\in\mathcal{X}$. The minimal Stinespring representation $\pi_{T^{\prime\prime}}$
of $T^{\prime\prime}$ is seen to be unitarily equivalent to
a subrepresentation of $\pi_{T_2}$.
For every $n\in\mathbb{N}$, $X_1,X_2,\cdots,X_n\in\mathcal{X}$ and
$Y_1,Y_2,\cdots,Y_n$ $\in\textrm{\boldmath $B$}(\mathcal{H})$,
\begin{align}
\sum_{i,j=1}^n Y_i^\ast T^{\prime\prime}(X_i^\ast X_j)Y_j
 &=\left(\sum_{i=1}^n \pi_T(X_i)WY_i\right)^\ast |(1-P)B|^2
\left(\sum_{j=1}^n \pi_T(X_j)WY_j\right) \nonumber\\
 &\leq \Vert(1-P)B\Vert^2\cdot \sum_{i,j=1}^n Y_i^\ast W^\ast\pi_T(X_i^\ast X_j)WY_j
\end{align}
i.e., $T^{\prime\prime}\leq \Vert(1-P)B\Vert^2\cdot W^\ast\pi_T(\cdot)W$.
Thus the minimal Stinespring representation of the CP map $W^\ast\pi_T(\cdot)W$
is unitarily equivalent to a subrepresentation of $\pi_{T_1}$,
so is $T^{\prime\prime}$. Hence $T^{\prime\prime}=0$,
equivalently, $0=\Vert T^{\prime\prime}\Vert=\Vert T^{\prime\prime}(1)\Vert=\Vert (1-P)BW\Vert^2$.
Since $B$ and $W$ are arbitrary, it holds that
$(1-P)BW=0$ for every $B\in\pi_{T}(\mathcal{X})^\prime$ and
$W\in\textrm{\boldmath $B$}(\mathcal{H},P\mathcal{K}_T)$.
This implies that $(1-P)BP=0$ for every $B\in\pi_{T}(\mathcal{X})^\prime$.
Thus $BP=PBP=PB$ for every $B\in\pi_{T}(\mathcal{X})^\prime$, that is to say,
$P\in\pi_{T}(\mathcal{X})^{\prime\prime}$.\\
$(2)\Rightarrow (1)$ Assume $(2)$. Then there exists no partial isometry $U\in\pi_T(\mathcal{X})^\prime$
such that $(1-P)UP=U$. Thus no subrepresentation of $\pi_{T_1}$ is unitarily equivalent to
that of $\pi_{T_2}$.
\end{proof}
\begin{definition}[]
$(S,\mathcal{F},\mu)$ is called a subcentral CP-measure space with barycenter
$T\in\mathrm{CP}(\mathcal{X},\mathcal{M})$ if its corresponding
von Neumann algebra is an abelian von Neumann subalgebra of
$\pi_T(\mathcal{X})^{\prime\prime}\cap\pi_T(\mathcal{X})^c$.
\end{definition}
\begin{theorem}[]
Let $(S,\mathcal{F},\mu)$ be a CP-measure space with  barycenter
$T\in\mathrm{CP}(\mathcal{X},\mathcal{M})$. The following conditions are equivalent:\\
$(1)$ For all $\Delta\in\mathcal{F}$,
\begin{equation}
\mu(\Delta)\;\hspace{0.3em}
\raisebox{-.41ex}{$\circ$}\hspace{-0.53em}\raisebox{0.7ex}{\rotatebox[origin=c]{-90}{--}}\hspace{0.5em}
\;\mu(\Delta^c);
\end{equation}
$(2)$ $(S,\mathcal{F},\mu)$ be a subcentral CP-measure space with  barycenter $T$.
\end{theorem}

\begin{proof}
 $\quad$\\
$(1)\Rightarrow (2)$ Let $\rho$ be a normal faithful state on $\mathcal{M}$ and
$\nu$ be a (scalar-valued) positive
finite measure on $(S,\mathcal{F})$ which is equivalent to $\rho\circ\mu$.
Assume $(1)$. By Lemma \ref{Sdisjoint}, $(S,\mathcal{F},\mu)$ is an orthogonal
CP-measure space of $T$ and the image $\kappa_\mu(f)$ of projections $f$ of $L^\infty(\nu)$ by $\kappa_\mu$
is projections of $\pi_T(\mathcal{X})^{\prime\prime}\cap\pi_T(\mathcal{X})^c$.
Since $\kappa_\mu$ is a $^\ast$-representation of $L^\infty(\nu)$ into $\pi_T(\mathcal{X})^c$,
$\kappa_\mu(L^\infty(\nu))$ is an abelian von Neumann subalgebra of
$\pi_T(\mathcal{X})^{\prime\prime}\cap\pi_T(\mathcal{X})^c$.\\
$(2)\Rightarrow (1)$ Assume $(2)$. For any $\Delta\in\mathcal{F}$, $\kappa_\mu(\chi_\Delta)$ and
$\kappa_\mu(\chi_{\Delta^c})$ are mutually orthogonal projections contained in
$\pi_T(\mathcal{X})^{\prime\prime}\cap\pi_T(\mathcal{X})^c$ such that
$\kappa_\mu(\chi_\Delta)+\kappa_\mu(\chi_{\Delta^c})=1$. By Lemma \ref{Sdisjoint} and
the formula $\mu(\Delta,X)=V_T^\ast \kappa_\mu(\chi_\Delta)\pi_T(X)V_T$, for all $X\in\mathcal{X}$ and
$\Delta\in\mathcal{F}$, it holds that $\mu(\Delta)\hspace{0.3em}
\raisebox{-.41ex}{$\circ$}\hspace{-0.53em}\raisebox{0.7ex}{\rotatebox[origin=c]{-90}{--}}\hspace{0.5em}
\mu(\Delta^c)$.
\end{proof}

\section{Local States and Generalized Local Sectors}\label{se:6}
DHR theory is a sector theory whose reference state is a vacuum $\omega_0$.
However, physically important reference states are not only vacuum states,
e.g., KMS states $\omega_\beta$, $\beta>0$ and relativistic KMS states \cite{Oj86,BB,Oj03,Oj04,Oj05}.
Thus we attempt to define the concept of sector for local states
to understand emergence processes in various representations.
The discussion here is, of course, applicable to DHR-DR theory.

In general, reference states are assumed to be translation-invariant.
We believe that this assumption is
very natural\footnote{On the other hand, the assumption of
Lorentz invariance for reference states is not always natural
because it is known that the invariance is broken in KMS states \cite{Oj86}.}.
Hence their GNS representations unitarily implement
the action $\alpha|_{\mathbb{R}^4}$ of the translation group $\mathbb{R}^4$.
We assume that reference representations $\pi$ which appear here are given by the GNS representations
$\pi_\varphi$ of translation-invariant states $\varphi\in E_\mathcal{A}$.
Thus $\pi$ is separable and unitarily implements the action $\alpha|_{\mathbb{R}^4}$
of $\mathbb{R}^4$.

Let $T$ be an element of $E_{\mathcal{A},\pi}^L(\Lambda)$
for some $\Lambda\in \mathcal{K}_\Subset^{DC}$ and $\theta_T$ be the commutant lifting of
$(V_T^\ast \pi_T(\mathcal{A})^{\prime\prime}V_T)^\prime$ into
$\{V_TV_T^\ast\}^\prime \cap \pi_T(\mathcal{A})^\prime$ (see Theorem \ref{CommLift}).
We are ready to define the concept of local sector for local states:
\begin{definition}[]
An element $T$ of $E_{\mathcal{A},\pi}^L(\Lambda)$ is called a local factor state
if a unique von Neumann algebra $\mathcal{Z}_{T,\mathrm{full}}:=
\pi_T(\mathcal{A})^{\prime\prime}\cap\pi_T(\mathcal{A})^c$ is equal to
$\theta_T(\mathfrak{Z}_\pi(\mathcal{A}))$, where $\mathfrak{Z}_\pi(\mathcal{A})$ is the
center $\pi(\mathcal{A})^{\prime}\cap\pi(\mathcal{A})^{\prime\prime}$
of the von Neumann algebra $\pi(\mathcal{A})^{\prime\prime}$.
We denote by $F_{\mathcal{A},\pi}^L(\Lambda)$ the set of local factor states on $\mathcal{A}$
with respect to $\pi$ with region $\Lambda$.
A local sector of $\mathcal{A}$ with respect to $\pi$ in $\Lambda$ is a quasi-equivalence class of
a local factor state $T\in F_{\mathcal{A},\pi}^L(\Lambda)$.
\end{definition}
\begin{remark}
The definition of local sector is a natural generalization of sectors for states to local states.
A sector of a C$^{\ast}$-algebra $\mathcal{X}$ is defined by a quasi-equivalent class of factor states
\cite{Oj03,Oj04,Oj05}.
A state $\omega$ on a C$^{\ast}$-algebra $\mathcal{X}$ is called a factor state if the center
$\mathfrak{Z}_{\omega}(\mathcal{X})$ of $\pi_{\omega}(\mathcal{X})^{\prime\prime}$ is trivial.
\end{remark}
The von Neumann algebra $\theta_T(\mathfrak{Z}_\pi(\mathcal{A}))\cong \mathfrak{Z}_\pi(\mathcal{A})$
describes bounded (measurable) functions on
order parameters of the system already emerged in a given reference representation $\pi$
and is common for all local states.
This is the reason why we should treat $\theta_T(\mathfrak{Z}_\pi(\mathcal{A}))$ as a coefficient algebra
in the definition of local sector for local states.

We can apply the method discussed in the previous section to local states.
In the context of local states, we are interested in sectors involved in a given local state $T$.
In the light of central decompositions of states,
we desire a subcentral CP-measure space of $T$ corresponding to
the abelian von Neumann algebra $\mathcal{Z}_T$ determined by the following two equations:
\begin{align}
\mathcal{Z}_T\vee \theta_T(\mathfrak{Z}_\pi(\mathcal{A})) &=\mathcal{Z}_{T,\mathrm{full}}, \\
\mathcal{Z}_T\cap \theta_T(\mathfrak{Z}_\pi(\mathcal{A}))&=\mathbb{C}1.
\end{align}
In contrast to $\theta_T(\mathfrak{Z}_\pi(\mathcal{A}))\cong \mathfrak{Z}_\pi(\mathcal{A})$,
the von Neumann algebra $\mathcal{Z}_T$ satisfying the above equations
describes bounded (measurable) functions on order parameters
of the system newly emerged in the given local state $T$.
In addition, the integral decomposition theory for separable representations of separable C$^\ast$-algebras
\cite{BR1,T79} can be applied to $(\pi_T,\mathcal{K}_T)$ 
if $\mathcal{A}$ is separable:
\begin{equation}
(\pi_T,\mathcal{K}_T)=\int^{\oplus}_Z (\pi_{T,z},\mathcal{K}_{T,z})\; d\nu_{\mathcal{Z}_T}(z),
\end{equation}
where $Z$ is a standard Borel space and $\nu_{\mathcal{Z}_T}$ is a Borel measure on $Z$.
We would like to emphasize that
a family $\{(\pi_{T,z},\mathcal{K}_{T,z})\}_{z\in Z}$ of representations
corresponds to that of mutually disjoint
local factor states\footnote{If a ``subcentral" version of \cite[Theorem 4.2]{Fujimoto} holds,
we can find candidates of local sectors involved in a given local state more precisely.}.

Local sector analysis of local states with respect to
translation-invariant representations other than vacuum ones will be discussed in \cite{O14no2}.

\section*{Acknowledgment}
One (K.O.) of the authors would like to thank Professor M. Ozawa for his encouragement and useful comments.
K.O. is supported by the John Templeton Foundations, No. 35771
and by the JSPS KAKENHI, No.~26247016.
I.O. is supported by the JSPS KAKENHI, No. 24320008.
H.S. is supported by the JSPS KAKENHI, No. 26870696.

\end{document}